\documentclass{article}
 
\usepackage{amsmath, amssymb, amsthm}
\usepackage{geometry}
\usepackage{fancyhdr}
\usepackage{graphicx}
\usepackage{booktabs}
\usepackage{multirow}
\usepackage{float}
\usepackage{xcolor}
\usepackage[ruled,vlined,linesnumbered]{algorithm2e}
\usepackage{hyperref}
\usepackage{natbib}
\usepackage{sectsty}   
 
\allsectionsfont{\scshape}
 
\fancypagestyle{plain}{%
  \fancyhf{}
  
  \fancyhead[L]{\small Accepted at the SeQureDB Workshop at ACM SIGMOD 2026 and TPDP Workshop 2026}
  \fancyfoot[C]{\thepage}
}
 
\newtheorem{theorem}{Theorem}[section]
\newtheorem{lemma}[theorem]{Lemma}
\newtheorem{proposition}[theorem]{Proposition}
\newtheorem{corollary}[theorem]{Corollary}

\theoremstyle{definition}
\newtheorem{definition}[theorem]{Definition}
\theoremstyle{remark}
\newtheorem{remark}[theorem]{Remark}
 
\begin{document}
 
{\LARGE\bfseries\scshape\raggedright
Refined Differentially Private Linear Regression via\\[4pt]
Extension of a Free Lunch Result\par}
 
\vspace{1.4em}
 
{\setlength{\parindent}{0pt}
\textbf{Sasmita Harini S}\\
Indian Institute of Science\\
Bangalore, India\\
\texttt{sasmitah@iisc.ac.in}
 
\vspace{0.8em}
\textbf{Anshoo Tandon}\\
CDPG, Indian Institute of Science\\
Bangalore, India\\
\texttt{anshoo.tandon@gmail.com}
}
 
\vspace{1.5em}
\thispagestyle{plain}
 
\begin{center}
{\sc Abstract}
\end{center}
\begin{quote}
As data-privacy regulations tighten and statistical models are increasingly deployed on sensitive human-sourced data, privacy-preserving linear regression has become a critical necessity. For the add-remove DP model, \citet{kulesza2024} and \citet{fitzsimons2024} have independently shown that the size of the dataset --- an important statistic for linear regression --- can be privately estimated for `free', via a simplex transformation of bounded variables and private sum queries on the transformed variables. In this work, we extend this free lunch result via carefully crafted multidimensional simplex transformations to variables and functions that are bounded in the interval $[0,1]$. We show that these transformations can be applied to refine the estimates of sufficient statistics needed for private simple linear regression based on ordinary least squares. We provide both analytical and numerical results to demonstrate the superiority of our approach. Our proposed transformations have general applicability and can be readily adapted for differentially private \emph{polynomial} regression.
\end{quote}
 
\vspace{0.5em}
 
\vspace{0.5em}
\section{Introduction}

Differential privacy \citep{dwork2006calibrating} has become the gold standard for releasing statistical insights from sensitive datasets while providing rigorous privacy guarantees \citep{NIST2025}. In linear regression, the challenge is to privately estimate model parameters from data $\mathcal{D} = \{(x_i, y_i)\}_{i=1}^n$ while maintaining utility. Traditional approaches \citep{alabi2022differentially} privatize sufficient statistics independently, adding noise scaled to individual sensitivities. Various techniques have been developed to improve upon basic Laplace mechanism, including concentrated differential privacy \citep{bun2016cdp}, smooth sensitivity methods \citep{bun2019average}, and robust estimation approaches \citep{kamath2020heavy,liu2021robust}. While effective, these methods suffer from high variance accumulation \citep{cai2021cost,kamath2023trilemma}, particularly in stringent privacy regimes or with small sample sizes.

Recent work by \citet{kulesza2024} and \citet{fitzsimons2024} revealed a key insight: for bounded data, the dataset size $n$ can be estimated ``for free'' via simplex transformations without additional privacy cost. In this paper, we extend this free lunch result to construct \emph{multiple independent unbiased estimators} for each sufficient statistic in simple linear regression based on ordinary least squares. By applying inverse-variance weighting, we achieve substantial variance reduction---up to $4.8\times$ for critical quadratic statistics while maintaining identical privacy guarantees.

Our contributions are: (1) a novel variance reduction technique (DP-RSS), that exploits algebraic structure in bounded data via complementary simplex transformations, to refine the estimates of private sufficient statistics required for ordinary least squares based linear regression, (2) theoretical analysis and empirical validation demonstrating enhanced performance over baseline methods, and (3) a natural extension to polynomial regression.

\section{Preliminaries}

\subsection{Notations}

The dataset is denoted by $D = \{(x_i, y_i)\}_{i=1}^n$, where $n$ denotes the total number of records (sample size).
We write $D \sim D'$ to denote neighboring datasets under the add/remove (unbounded) differential privacy model.
The regression parameters are denoted by $\alpha$ (slope) and $\beta$ (intercept), and $\varepsilon > 0$ represents the total privacy budget.
We use $\mathrm{Lap}(b)$ to denote the zero-mean Laplace distribution with scale parameter $b$, having probability density function $\frac{1}{2b}\exp\!\left(-\frac{|x|}{b}\right)$. For a vector $x \in \mathbb{R}^n$, the $\ell_p$-norm is defined as $\|x\|_p = \left( \sum_{i=1}^{n} |x_i|^p \right)^{1/p}$, $p \in [1, \infty]$. The $\ell_1$-sensitivity of a function $f$ is denoted by $\Delta := \max_{D \sim D'} \|f(D) - f(D')\|_1$.
We use $S_{g(x,y)}$ to denote the true sum statistic (e.g., $S_{x^2} = \sum_i x_i^2$), $\tilde{S}_{g(x,y)}$ to denote its noisy version obtained by adding appropriate Laplace noise, and $\hat{S}_{g(x,y)}$ to denote the refined estimator obtained via weighted averaging.
Finally, $w_i$ and $w'_i$ denote the optimal weights used for variance reduction.

\subsection{Differential Privacy}

\begin{definition}[Differential Privacy]
A randomized mechanism $\mathcal{M}$ that maps the output of a query function $F(D)$ on a dataset $D$ to a randomized output satisfies $(\varepsilon, \delta)$-differential privacy if for all pairs of neighboring datasets $D$ and $D'$ differing in at most one record (denoted $D \sim D'$), and for all measurable subsets $S$ of the output space:
\[
\Pr[\mathcal{M}(D) \in S] \leq e^{\varepsilon} \Pr[\mathcal{M}(D') \in S] + \delta.
\]
Here, $S$ denotes any event in the output space of $\mathcal{M}$.
\end{definition}

The privacy parameters $(\varepsilon, \delta)$ characterize the strength of the guarantee.
The parameter $\varepsilon > 0$ is called the \emph{privacy loss budget} and controls the multiplicative distance between output distributions under neighboring datasets.
A smaller $\varepsilon$ provides stronger privacy.
The parameter $\delta \geq 0$ allows a small probability of failure of the strict privacy guarantee and captures the probability with which the bound may be violated. When $\delta = 0$, the guarantee is referred to as \emph{pure differential privacy}.
In this work, we focus on $(\varepsilon, 0)$-differential privacy (pure DP).

\paragraph{Neighboring Datasets.}
We adopt the \emph{add-remove} model of adjacency, where two datasets $D$ and $D'$ are considered neighboring ($D \sim D'$) if one can be obtained from the other by adding or removing a single record. To achieve differential privacy, randomized mechanisms typically add carefully calibrated noise to the true query output.
The scale of the noise depends on the \emph{sensitivity} of the query function, which measures how much the query output can change between neighboring datasets.

For a function $F : \mathcal{X} \rightarrow \mathbb{R}^n$, the $\ell_p$-sensitivity is defined as:
\[
\Delta_p(F) := \sup_{D \sim D'} \|F(D) - F(D')\|_p.
\]

We use the $\ell_1$-sensitivity under the add-remove adjacency model. This choice aligns with mechanisms such as the Laplace mechanism, where noise is calibrated proportional to $\Delta_1(F)/\varepsilon$ to guarantee $(\varepsilon,0)$-differential privacy. We operate in the \emph{central} model of differential privacy \citep{NIST2025}, where a trusted curator holds the full dataset and releases a privatized output.

\subsection{Problem Statement}

Consider a simple linear regression model: $y_i = \alpha x_i + \beta + e_i$ where $e_i$ represents i.i.d.\ Gaussian noise with unknown variance. Similar to \citet{alabi2022differentially,chetty2018}, we assume that the variables have been scaled and shifted to satisfy $0 \leq x_i, y_i \leq 1$, for $1 \leq i \leq n$ (see Appendix~\ref{app:general_bounds} for the adaptation of DP-RSS to the general case where $x_{\min} \leq x_i \leq x_{\max}$ and $y_{\min} \leq y_i \leq y_{\max}$). Given a dataset $D = \{(x_i, y_i)\}_{i=1}^n$ and a privacy budget $\varepsilon > 0$, our goal is to estimate $\alpha$ and $\beta$ while satisfying $\varepsilon$-DP under the add/remove DP model. Let $\hat{\alpha}$ and $\hat{\beta}$ denote the private estimates of $\alpha$ and $\beta$, respectively. We measure the accuracy of the estimated linear model via evaluation of the mean absolute error (MAE) and the mean squared error (MSE), defined as
\begin{align}
L_1 \text{ Error: MAE} &= \mathbb{E} \left[ \left|(\hat{\alpha} x + \hat{\beta}) - (\alpha x + \beta) \right| \right] \label{eq:mae} \\
L_2 \text{ Error: MSE} &= \mathbb{E} \left[ \left( (\hat{\alpha} x + \hat{\beta}) - (\alpha x + \beta) \right)^2 \right] \label{eq:mse}
\end{align}

\section{Background and Related Work}
\label{sec:background}

\subsection{The Free Lunch via Simplex Transformations}

For univariate bounded data $x \in [0,1]$, the simplex transformation $x \mapsto (x, 1-x)$ enables a remarkable result \citep{kulesza2024,fitzsimons2024}: privatizing the sum pair $(S_x, S_{1-x})$ with Laplace noise $\mathrm{Lap}(1/\varepsilon)$ yields both a noisy sum $\tilde{S}_x$ and a ``free'' size estimate $\tilde{n} = (\tilde{S}_x + \tilde{S}_{1-x})$. The variance of $\tilde{S}_x$ remains identical to direct privatization, so $\tilde{n}$ incurs zero additional privacy cost.

This generalizes to any bounded function $f(x) \in [0,1]$: the sum pair $(S_f, S_{1-f})$ has joint $\ell_1$-sensitivity $1$, enabling private size estimation. The key is that the complementary components sum to a constant $n$, creating an algebraic constraint that can be exploited via post-processing.

\subsection{Prior Work on Private Linear Regression}

Traditional DP linear regression methods privatize sufficient statistics independently. For simple linear regression requiring $\{S_x, S_y, S_{x^2}, S_{xy}\}$, standard approaches allocate privacy budget equally and add independent noise to each statistic \citep{alabi2022differentially}. Under the \emph{swap} model (replacing one record with another), \citet{alabi2022differentially} achieved strong results by exploiting the known dataset size. In another related work, \citet{amin2023easyDP} suggested to partition the dataset into roughly $m=1000$ subsets, non-privately estimate a regression model on each, and then apply the exponential mechanism to privately estimate a model; however, such an approach may not be feasible for small to moderate sized datasets. A related line of work is explored in \citet{ferrando2025private}, who proposed private regression via data-dependent sufficient statistic perturbation; their approach, however, is restricted only to discrete data whereas our method operates directly on continuous data in $[0,1]$.

Another prominent approach is Differentially Private Stochastic Gradient Descent (DP-SGD) \citep{bassily2014private,abadi2016deep,cai2021cost}. These approaches typically require careful hyperparameter tuning, and some open-source implementations have been found to have overestimated their privacy guarantees \citep{annamalai2024shuffle,chua2024howprivate}.

The \emph{add/remove} model where neighboring datasets differ by addition or removal of a record imposes stricter requirements. Dataset size becomes unknown and must itself be estimated privately. This model is more challenging but provides stronger privacy guarantees \citep{NIST2025}. For comparison with prior art, we adapt the two baseline methods to the add/remove setting: DP-SS (Differentially Private Sufficient Statistics) \citep{alabi2022differentially} and DP-Theil-Sen mechanisms \citep{dwork2009RobustStatistics,alabi2022differentially}.

In this paper, we extend the free lunch insight for add/remove model beyond simple size estimation. We show that carefully designed multidimensional simplex transformations enable constructing \emph{multiple independent estimators} for the same statistic, which can be optimally combined to dramatically reduce variance without consuming additional privacy budget.

\subsection{Our Adaptations of Prior Work for Add/Remove DP}

The add/remove DP model poses stricter requirements as the dataset size is private. We adapt the DP-SS and DP-Theil-Sen algorithms as follows:

\subsubsection{DP-SS: Independent Sufficient Statistics}

Our DP-SS baseline represents the adapted version, based on \citet{kulesza2024,fitzsimons2024}, for add/remove DP. Key modifications, detailed in Algorithm~\ref{alg:dp_ss}, include:
\begin{itemize}
  \item \textbf{Private size estimation}: Since the dataset size $n$ is unknown, we apply simplex transformations to obtain complementary statistics $(S_x, S_{1-x})$, $(S_y, S_{1-y})$, $(S_{x^2}, S_{1-x^2})$, and $(S_{xy}, S_{1-xy})$. Each pair has $\ell_1$-sensitivity 1 in the add/remove DP model. Private estimation of the dataset size, $\tilde{n}$, is obtained via step~4 in Algorithm~\ref{alg:dp_ss}.
  \item \textbf{Equal budget split}: Privacy budget $\varepsilon$ is split equally as $\varepsilon' = \varepsilon/4$ across the four statistic pairs.
\end{itemize}

\begin{remark}[DP-SS Baseline Variance]
In DP-SS, the privacy budget $\varepsilon$ is split equally into $\varepsilon/4$ for each of $S_x, S_y, S_{x^2}, S_{xy}$. Each has sensitivity 1, so noise scale is $4/\varepsilon$, yielding $\mathrm{Var}(\mathrm{Lap}(4/\varepsilon)) = 2 \cdot (4/\varepsilon)^2 = 32/\varepsilon^2$. For the private sample size estimate, we have $\mathrm{Var}(\tilde{n}) = 16/\varepsilon^2$.
\end{remark}

\begin{algorithm}[H]
\caption{DP-SS: Differentially Private Sufficient Statistics for add/remove DP}
\label{alg:dp_ss}
\SetKwInOut{Dataset}{Dataset}
\SetKwInOut{Params}{Privacy params}
\Dataset{$\{(x_i, y_i)\}_{i=1}^n \in ([0, 1] \times [0, 1])^n$}
\Params{$\varepsilon > 0$}
\BlankLine
$\varepsilon' = \varepsilon / 4$, $\Delta = 1$\;
\tcp{Sample Laplace Noise and Add to Statistics}
$L_{S_k} \sim \mathrm{Lap}(0, \Delta/\varepsilon')$, where $1 \le k \le 8$\;
$\tilde{S}_{x} = \sum x_i + L_{S_1}$\;
$\tilde{S}_{1-x} = \sum (1-x_i) + L_{S_2}$\;
$\tilde{S}_{y} = \sum y_i + L_{S_3}$\;
$\tilde{S}_{1-y} = \sum (1-y_i) + L_{S_4}$\;
$\tilde{S}_{xy} = \sum x_i y_i + L_{S_5}$\;
$\tilde{S}_{1-xy} = \sum (1-x_i y_i) + L_{S_6}$\;
$\tilde{S}_{x^2} = \sum x_i^2 + L_{S_7}$\;
$\tilde{S}_{1-x^2} = \sum (1-x_i^2) + L_{S_8}$\;
$\tilde{n} = (\tilde{S}_{x} + \tilde{S}_{1-x} + \tilde{S}_{y} + \tilde{S}_{1-y} + \tilde{S}_{xy} + \tilde{S}_{1-xy} + \tilde{S}_{x^2} + \tilde{S}_{1-x^2}) / 4$\;
\If{$\tilde{n} \le 0$}{\Return $(\hat{\alpha}, \hat{\beta}) = (0, 0.5)$\;}
\BlankLine
$\widetilde{ncov} = \tilde{S}_{xy} - (\tilde{S}_{x} \cdot \tilde{S}_{y}) / \tilde{n}$\;
$\widetilde{nvar} = \tilde{S}_{x^2} - (\tilde{S}_{x}^2) / \tilde{n}$\;
\If{$\widetilde{nvar} \le 0$}{\Return $(\hat{\alpha}, \hat{\beta}) = (0, 0.5)$\;}
\BlankLine
$\hat{\alpha} = \widetilde{ncov} / \widetilde{nvar}$\;
$\hat{\beta} = (\tilde{S}_{y} - \hat{\alpha} \cdot \tilde{S}_{x}) / \tilde{n}$\;
\BlankLine
\Return{$\hat{\alpha}, \hat{\beta}$}
\end{algorithm}

\subsubsection{DP-Theil-Sen: Robust Median-Based Estimation}

We adapt the DP-Theil-Sen matching-based algorithm in \citet{alabi2022differentially} to the add/remove DP model. The adapted version is detailed in Algorithm~\ref{alg:dp_theil_sen}, and includes a private estimation of the dataset size (see step~3), that requires privacy budget $\varepsilon' = \varepsilon/3$.

\begin{algorithm}[H]
\caption{DP-Theil-Sen with Random Matchings for add/remove DP}
\label{alg:dp_theil_sen}
\SetKwInOut{Data}{Data}
\SetKwInOut{Params}{Privacy params}
\SetKwInOut{Hyper}{Hyperparams}
\Data{$\{(x_i, y_i)\}_{i=1}^n \in ([0, 1] \times [0, 1])^n$}
\Params{$\varepsilon > 0$}
\Hyper{$k$ iterations, Clipping bounds $[r_L, r_U]$}
\BlankLine
$\varepsilon' = \varepsilon / 3$\;
Sample $L_n \sim \mathrm{Lap}(0, 1/\varepsilon')$\;
$\tilde{n} = n + L_n$\;
$\mathbf{z}^{(p25)}, \mathbf{z}^{(p75)} = [\,]$\;
\BlankLine
\For{$round = 1$ \KwTo $k$}{
  $M \leftarrow \mathrm{RandomMatching}(n)$\;
  \For{each pair $(i, j) \in M$}{
    \If{$x_j \neq x_i$}{
      $s = (y_j - y_i) / (x_j - x_i)$\;
      $mid_x = (x_i + x_j) / 2$\;
      $mid_y = (y_i + y_j) / 2$\;
      $z_{j,l}^{(p25)} = s \cdot (0.25 - mid_x) + mid_y$\;
      $z_{j,l}^{(p75)} = s \cdot (0.75 - mid_x) + mid_y$\;
      Append $z_{j,l}^{(p25)}$ to $\mathbf{z}^{(p25)}$\;
      Append $z_{j,l}^{(p75)}$ to $\mathbf{z}^{(p75)}$\;
    }
  }
}
\BlankLine
$\tilde{p}_{25} = \mathrm{DPmed}(\mathbf{z}^{(p25)}, \varepsilon', (\tilde{n}, k, \mathrm{hyperparams}))$\;
$\tilde{p}_{75} = \mathrm{DPmed}(\mathbf{z}^{(p75)}, \varepsilon', (\tilde{n}, k, \mathrm{hyperparams}))$\;
\BlankLine
$\hat{\alpha} = (\tilde{p}_{75} - \tilde{p}_{25}) / 0.5$\;
$\hat{\beta} = \tilde{p}_{25} - \hat{\alpha} \cdot 0.25$\;
\Return{$\hat{\alpha}, \hat{\beta}$}
\end{algorithm}

\begin{remark}[DP Median Mechanism]
We reuse the exponential mechanism for private median computation, presented as $\mathrm{DPmed}(\cdot)$ algorithm in \citet{alabi2022differentially}, for our adapted DP-Theil-Sen algorithm.
\end{remark}

\begin{remark}[Experimental Setup]
For experimental evaluation of DP-Theil-Sen, we used the following hyperparameters for $\mathrm{DPmed}(\cdot)$: $k=1$ and $r_L=-2$ and $r_U=2$.
\end{remark}

\section{DP-RSS: Differentially Private Refined Sufficient Statistics}

In this section, we present our main contribution, DP-RSS: Differentially Private Refined Sufficient Statistics, a novel mechanism that exploits algebraic structure in bounded data and applies carefully crafted multidimensional simplex transformations that substantially lower the variance of the required sufficient statistics with the same privacy budget.

Let $\varepsilon > 0$ denote the total privacy budget. We partition this budget into two equal parts: $\varepsilon_1 = \varepsilon_2 = \frac{\varepsilon}{2}$. The first budget $\varepsilon_1$ is allocated to statistics involving only the $x$-coordinates of the data, while the second budget $\varepsilon_2$ is allocated to statistics involving both $x$ and $y$ coordinates (joint statistics). This separation is crucial for our variance reduction technique, as it allows us to construct independent noisy estimators that can be combined via weighted averaging.

\subsection{Extended Free Lunch}

The core idea behind DP-RSS is to exploit the algebraic structure of bounded data through carefully designed simplex transformations. For any bounded variable or function mapping to $[0,1]$, we can decompose it into complementary components that sum to a constant, thereby enabling us to construct multiple independent estimators from a single privacy budget allocation.

Recall that $\{S_x, S_y, S_{x^2}, S_{xy}, n\}$ constitute sufficient statistics for linear regression. Towards obtaining private estimates of these expressions, we project $x$ and $xy$ onto the simplex (the space of non-negative values that sum to $1$) as follows:
\begin{align*}
x_i &\mapsto \left( x_i^2,\; x_i - x_i^2,\; 1 - x_i \right), \\
(x_i, y_i) &\mapsto \left( x_i y_i,\; (1-x_i)y_i,\; 1-y_i \right).
\end{align*}
Summing these transformations over the entire dataset results in the following vectors of statistics:
\begin{align}
&(S_{x^2}, \, S_{x-x^2}, \, S_{1-x}), \label{eq:sum_x} \\
&(S_{xy}, \, S_{(1-x)y}, \, S_{1-y}). \label{eq:sum_xy}
\end{align}

\paragraph{Sensitivity analysis and the key advantage.} A crucial property of these multidimensional simplex transformations is that each of the two vectors in~\eqref{eq:sum_x} and~\eqref{eq:sum_xy} has $\ell_1$-sensitivity equal to $1$, since each record contributes a triplet summing to $1$. The key insight is that adding or removing a single record $(x_i, y_i) \in [0,1]^2$ changes each vector by exactly $(x_i^2, x_i - x_i^2, 1-x_i)$ or $(x_i y_i, (1-x_i)y_i, 1-y_i)$ respectively, and these triplets have $\ell_1$-norm equal to 1 due to the simplex constraint.

Importantly, if we were to privatize individual statistics like $S_{x^2}$ or $S_{xy}$ independently without the simplex structure, each would \emph{also} have sensitivity 1. However, privatizing them separately would require partitioning the privacy budget across multiple independent mechanisms, resulting in significantly higher noise variance. By grouping them via simplex transformations, we can privatize entire groups with a single noise addition per component, and then exploit algebraic redundancy to construct multiple estimators.

These independent estimators, derived from separately privatized groups, can then be optimally combined via inverse-variance weighting to achieve substantial variance reduction without consuming any additional privacy budget. This is the essence of our extension of the ``free lunch'' phenomenon: the simplex structure provides algebraic redundancy that can be exploited for variance reduction through post-processing alone, which is a consequence of the post-processing immunity property of differential privacy.

Algorithm~\ref{alg:dp_rss_compact} presents the complete DP-RSS mechanism, which implements these simplex transformations, adds calibrated Laplace noise to the transformed statistics, and recovers the private regression parameters $(\hat{\alpha},\hat{\beta})$ using optimally weighted refined estimators. We remark that DP-RSS can be readily extended to private polynomial regression (see Section~\ref{sec:polynomial} for a detailed sketch), and that the application of our simplex transformation technique to other problem scenarios may be of independent interest.

We now formally establish the theoretical foundations of DP-RSS. In what follows, we present definitions, lemmas, theorems, and propositions that rigorously characterize the privacy guarantees, variance properties, and optimality of our approach.

\subsection{Noisy Statistics and Their Privacy Guarantees}

Our goal is to privately estimate fundamental statistics such as $S_x = \sum_i x_i$, $S_y = \sum_i y_i$, $S_{x^2} = \sum_i x_i^2$, and $S_{xy} = \sum_i x_i y_i$. A naive approach would add independent Laplace noise to each statistic, but this ignores algebraic relationships that can be exploited for variance reduction. Our key insight is to design \emph{redundant} statistics whose sensitivities sum to exactly 1, enabling us to extract multiple independent estimates of the same quantity and combine them optimally.

All proofs for results in this section are provided in Appendix~\ref{app:proofs}.

\begin{algorithm}[H]
\caption{DP-RSS for Linear Regression}
\label{alg:dp_rss_compact}
\SetKwInOut{Input}{Input}
\SetKwInOut{Output}{Output}
\SetKwInOut{PP}{Privacy}
\Input{$D=\{(x_i,y_i)\}_{i=1}^n$, $(x_i,y_i)\in[0,1]^2$}
\PP{$\varepsilon>0$}
\Output{$(\hat\alpha,\hat\beta)$ or $\perp$}
\BlankLine
\tcp{1. Define privacy guarantees and initialise variables to 0}
$\varepsilon_1=\varepsilon_2=\varepsilon/2,\;\Delta=1$\;
$S_{x^2}=S_{x-x^2}=S_{1-x}=S_{xy}=S_{(1-x)y}=S_{1-y}=0$\;
\BlankLine
\tcp{2. Compute clean statistics}
\For{$(x_i,y_i) \in D$}{
  $S_{x^2} \mathrel{+}= x_i^2$;\quad
  $S_{x-x^2} \mathrel{+}= (x_i - x_i^2)$;\quad
  $S_{1-x} \mathrel{+}= (1-x_i)$\;
  $S_{xy} \mathrel{+}= x_i y_i$;\quad
  $S_{(1-x)y} \mathrel{+}= (1-x_i)y_i$;\quad
  $S_{1-y} \mathrel{+}= (1-y_i)$\;
}
\BlankLine
\tcp{3. Add Laplace noise}
\For{$j\in\{x^2,\,x-x^2,\,1-x\}$}{
  $\tilde{S}_j = S_j + \mathrm{Lap}(0,\,\Delta/\varepsilon_1)$\;
}
\For{$j\in\{xy,\,(1-x)y,\,1-y\}$}{
  $\tilde{S}_j = S_j + \mathrm{Lap}(0,\,\Delta/\varepsilon_2)$\;
}
\BlankLine
\tcp{4. Noisy sample sizes}
$\tilde{n}_x = \tilde{S}_{x^2}+\tilde{S}_{x-x^2}+\tilde{S}_{1-x}$\;
$\tilde{n}_y = \tilde{S}_{xy}+\tilde{S}_{(1-x)y}+\tilde{S}_{1-y}$\;
$\tilde{n} = (\tilde{n}_x+\tilde{n}_y)/2$\;
\If{$\tilde{n}\le 0$}{\Return $(\hat{\alpha}, \hat{\beta}) = (0, 0.5)$\;}
\BlankLine
\tcp{5. Refined estimators}
$\hat{S}_{x^2}=\tfrac{5}{6}\tilde{S}_{x^2}+\tfrac{1}{6}(\tilde{n}_y-\tilde{S}_{x-x^2}-\tilde{S}_{1-x})$\;
$\hat{S}_{xy}=\tfrac{5}{6}\tilde{S}_{xy}+\tfrac{1}{6}(\tilde{n}_x-\tilde{S}_{(1-x)y}-\tilde{S}_{1-y})$\;
$\hat{S}_x=\tfrac{2}{3}(\tilde{S}_{x^2}+\tilde{S}_{x-x^2})+\tfrac{1}{3}(\tilde{n}_y-\tilde{S}_{1-x})$\;
$\hat{S}_y=\tfrac{2}{3}(\tilde{S}_{xy}+\tilde{S}_{(1-x)y})+\tfrac{1}{3}(\tilde{n}_x-\tilde{S}_{1-y})$\;
\BlankLine
$\mathrm{det}=\hat{S}_{x^2}\tilde{n}-\hat{S}_x^2$\;
\If{$\mathrm{det} \le 0$}{\Return $(\hat{\alpha}, \hat{\beta}) = (0, 0.5)$\;}
$\hat{\alpha}=(\tilde{n}\hat{S}_{xy}-\hat{S}_x\hat{S}_y)\,/\,\mathrm{det}$\;
$\hat{\beta}=(\hat{S}_{x^2}\hat{S}_y-\hat{S}_x\hat{S}_{xy})\,/\,\mathrm{det}$\;
\BlankLine
\Return $(\hat{\alpha},\hat{\beta})$
\end{algorithm}

\subsubsection{Group 1: Statistics Based on $x$-Coordinates}

We begin by partitioning the contribution of each $x$-coordinate into three non-negative components. For any $x \in [0,1]$, observe that:
$$
x^2 + (x - x^2) + (1-x) = 1.
$$
This identity is the foundation of our approach: when we sum over all data points, the three resulting statistics will sum to $n$, creating useful redundancy.

\begin{definition}[Group 1 Statistics]
\label{def:group1}
For a dataset $D = \{(x_i, y_i)\}_{i=1}^n$ with $x_i \in [0,1]$, define:
\begin{equation}
S_{x^2} := \sum_{i=1}^{n} x_i^2, \qquad S_{x-x^2} := \sum_{i=1}^{n} (x_i - x_i^2), \qquad S_{1-x} := \sum_{i=1}^{n} (1 - x_i).
\end{equation}
These statistics satisfy the constraint $S_{x^2} + S_{x-x^2} + S_{1-x} = n$.
\end{definition}

\begin{lemma}[Sensitivity of Group 1]
\label{lem:sensitivity_group1}
The function $f(D) = (S_{x^2}, S_{x-x^2}, S_{1-x})$ has global $\ell_1$-sensitivity $\Delta_1(f) = 1$.
\end{lemma}

This sensitivity of exactly 1 is crucial as we can privatize \emph{all three} statistics using the same noise scale that would be required to privatize just \emph{one} statistic with sensitivity 1. We add independent Laplace noise to each component:

\begin{definition}[Noisy Group 1 Statistics]
\label{def:noisy_group1}
Let $Z_{11}, Z_{12}, Z_{13} \stackrel{\mathrm{iid}}{\sim} \mathrm{Lap}(1/\varepsilon_1)$. Define:
\begin{equation}
\tilde{S}_{x^2} := S_{x^2} + Z_{11}, \qquad \tilde{S}_{x-x^2} := S_{x-x^2} + Z_{12}, \qquad \tilde{S}_{1-x} := S_{1-x} + Z_{13}.
\end{equation}
\end{definition}

\begin{corollary}[Privacy of Group 1]
\label{cor:privacy_group1}
The release of $(\tilde{S}_{x^2}, \tilde{S}_{x-x^2}, \tilde{S}_{1-x})$ satisfies $\varepsilon_1$-differential privacy.
\end{corollary}

\subsubsection{Group 2: Statistics Based on $(x,y)$-Coordinates}

We apply the same principle to statistics involving both coordinates. For any $(x,y) \in [0,1]^2: xy + (1-x)y + (1-y) = y + (1-y) = 1.$
Summing over the dataset yields three statistics that sum to $n$.

\begin{definition}[Group 2 Statistics]
\label{def:group2}
For a dataset $D = \{(x_i, y_i)\}_{i=1}^n$ with $x_i, y_i \in [0,1]$, define:
\begin{equation}
S_{xy} := \sum_{i=1}^{n} x_i y_i, \qquad S_{(1-x)y} := \sum_{i=1}^{n} (1 - x_i) y_i, \qquad S_{1-y} := \sum_{i=1}^{n} (1 - y_i).
\end{equation}
These statistics satisfy the constraint $S_{xy} + S_{(1-x)y} + S_{1-y} = n$.
\end{definition}

\begin{lemma}[Sensitivity of Group 2]
\label{lem:sensitivity_group2}
The function $g(D) = (S_{xy}, S_{(1-x)y}, S_{1-y})$ has global $\ell_1$-sensitivity $\Delta_1(g) = 1$.
\end{lemma}

\begin{definition}[Noisy Group 2 Statistics]
\label{def:noisy_group2}
Let $Z_{21}, Z_{22}, Z_{23} \stackrel{\mathrm{iid}}{\sim} \mathrm{Lap}(1/\varepsilon_2)$, independent of Group 1 noise. Define:
\begin{equation}
\tilde{S}_{xy} := S_{xy} + Z_{21}, \qquad \tilde{S}_{(1-x)y} := S_{(1-x)y} + Z_{22}, \qquad \tilde{S}_{1-y} := S_{1-y} + Z_{23}.
\end{equation}
\end{definition}

\begin{corollary}[Privacy of Group 2]
\label{cor:privacy_group2}
The release of $(\tilde{S}_{xy}, \tilde{S}_{(1-x)y}, \tilde{S}_{1-y})$ satisfies $\varepsilon_2$-differential privacy.
\end{corollary}

\subsubsection{Overall Privacy Guarantee}

Since the two groups use independent noise, the standard composition theorem applies:

\begin{theorem}[Privacy Composition]
\label{thm:privacy_composition}
The joint release $\mathcal{M}(D) = \bigl( \tilde{S}_{x^2}, \tilde{S}_{x-x^2}, \tilde{S}_{1-x}, \tilde{S}_{xy}, \tilde{S}_{(1-x)y}, \tilde{S}_{1-y} \bigr)$ satisfies $\varepsilon$-differential privacy, where $\varepsilon = \varepsilon_1 + \varepsilon_2$.
\end{theorem}

In what follows, we set $\varepsilon_1 = \varepsilon_2 = \varepsilon/2$ for balanced allocation.

\subsection{Post-Processing: Constructing Refined Estimators}

We now show how the redundancy built into our statistics enables variance reduction. The key observation is that we can construct \emph{multiple independent unbiased estimators} of the same target quantity, then combine them optimally. Since differential privacy is immune to post-processing, all operations in this section preserve our $\varepsilon$-DP guarantee.

\subsubsection{Recovering the Sample Size}

The constraints $S_{x^2} + S_{x-x^2} + S_{1-x} = n$ and $S_{xy} + S_{(1-x)y} + S_{1-y} = n$ allow us to estimate $n$ in two independent ways:

\begin{definition}[Private Sample Size]
\label{def:noisy_n}
Define:
\begin{align}
\tilde{n}_x &:= \tilde{S}_{x^2} + \tilde{S}_{x-x^2} + \tilde{S}_{1-x}, \qquad \tilde{n}_y := \tilde{S}_{xy} + \tilde{S}_{(1-x)y} + \tilde{S}_{1-y}, \\
\tilde{n} &:= \frac{\tilde{n}_x + \tilde{n}_y}{2}.
\end{align}
\end{definition}

Both $\tilde{n}_x$ and $\tilde{n}_y$ are unbiased estimators of $n$: adding Laplace noise with mean zero preserves the sum's expectation. Crucially, $\tilde{n}_x$ depends only on Group 1 noise $(Z_{11}, Z_{12}, Z_{13})$, while $\tilde{n}_y$ depends only on Group 2 noise $(Z_{21}, Z_{22}, Z_{23})$. This independence will be essential for constructing refined estimators.

\subsubsection{The Principle of Optimal Combination}

When we have two independent unbiased estimators of the same quantity, classical statistics tells us how to combine them optimally:

\begin{lemma}[Optimal Weighted Average]
\label{lem:optimal_weights}
Let $\hat{\theta}_1, \hat{\theta}_2$ be independent unbiased estimators of $\theta$ with variances $\sigma_1^2, \sigma_2^2$. The minimum-variance unbiased combination $\hat{\theta} = w_1\hat{\theta}_1 + w_2\hat{\theta}_2$ with $w_1 + w_2 = 1$ uses weights inversely proportional to variance:
$w_1^* = \frac{\sigma_2^2}{\sigma_1^2 + \sigma_2^2}, w_2^* = \frac{\sigma_1^2}{\sigma_1^2 + \sigma_2^2},$
achieving variance $\mathrm{Var}(\hat{\theta}^*) = \bigl(\sigma_1^{-2} + \sigma_2^{-2}\bigr)^{-1} < \min(\sigma_1^2, \sigma_2^2)$.
\end{lemma}

The final inequality shows that optimal combination \emph{always} improves upon either estimator alone. We now apply this principle systematically.

\subsubsection{Refined Estimator for $S_{x^2}$}

The statistic $S_{x^2}$ can be estimated in two ways. First estimator (direct): Simply use the noisy version $\tilde{S}_{x^2}$, and second estimator (indirect): From the constraint $S_{x^2} + S_{x-x^2} + S_{1-x} = n$, we have $S_{x^2} = n - (S_{x-x^2} + S_{1-x})$. Using Group 2's estimate of $n$:
$\hat{S}_{x^2}^{(2)} := \tilde{n}_y - (\tilde{S}_{x-x^2} + \tilde{S}_{1-x}).$

The crucial observation is that these two estimators are \emph{independent}: $\tilde{S}_{x^2}$ depends only on $Z_{11}$, while $\hat{S}_{x^2}^{(2)}$ depends on the disjoint set $\{Z_{12}, Z_{13}, Z_{21}, Z_{22}, Z_{23}\}$.

\begin{proposition}[Two Estimators for $S_{x^2}$]
\label{prop:estimators_x2}
The estimators $\hat{S}_{x^2}^{(1)} := \tilde{S}_{x^2}$ and $\hat{S}_{x^2}^{(2)} := \tilde{n}_y - (\tilde{S}_{x-x^2} + \tilde{S}_{1-x})$ are independent unbiased estimators of $S_{x^2}$ with variances $8/\varepsilon^2$ and $40/\varepsilon^2$, respectively.
\end{proposition}

Applying Lemma~\ref{lem:optimal_weights}, the optimal weights are $5/6$ and $1/6$ (proportional to $40$ and $8$, i.e., inversely proportional to variances):

\begin{theorem}[Refined Estimator for $S_{x^2}$]
\label{thm:refined_x2}
Define:
$$\hat{S}_{x^2} := \frac{5}{6}\,\tilde{S}_{x^2} + \frac{1}{6}\bigl[\tilde{n}_y - (\tilde{S}_{x-x^2} + \tilde{S}_{1-x})\bigr].$$
Then $\hat{S}_{x^2}$ is unbiased with variance $\mathrm{Var}(\hat{S}_{x^2}) = \dfrac{20}{3\varepsilon^2} \approx \dfrac{6.67}{\varepsilon^2}$.
\end{theorem}

This represents a $17\%$ variance reduction compared to using $\tilde{S}_{x^2}$ alone (variance $8/\varepsilon^2$).

\subsubsection{Refined Estimator for $S_{xy}$}

By symmetry, we can estimate $S_{xy}$ in two ways: directly via $\tilde{S}_{xy}$, or indirectly using Group 1's estimate of $n$ combined with the Group 2 constraint:

\begin{proposition}[Two Estimators for $S_{xy}$]
\label{prop:estimators_xy}
The estimators $\hat{S}_{xy}^{(1)} := \tilde{S}_{xy}$ and $\hat{S}_{xy}^{(2)} := \tilde{n}_x - (\tilde{S}_{(1-x)y} + \tilde{S}_{1-y})$ are independent unbiased estimators of $S_{xy}$ with variances $8/\varepsilon^2$ and $40/\varepsilon^2$, respectively.
\end{proposition}

Applying Lemma~\ref{lem:optimal_weights}, the optimal weights are $5/6$ and $1/6$:

\begin{theorem}[Refined Estimator for $S_{xy}$]
\label{thm:refined_xy}
Define:
$$\hat{S}_{xy} := \frac{5}{6}\,\tilde{S}_{xy} + \frac{1}{6}\bigl[\tilde{n}_x - (\tilde{S}_{(1-x)y} + \tilde{S}_{1-y})\bigr].$$
Then $\hat{S}_{xy}$ is unbiased with variance $\mathrm{Var}(\hat{S}_{xy}) = \dfrac{20}{3\varepsilon^2}$.
\end{theorem}

\subsubsection{Refined Estimator for $S_x$}

For $S_x = \sum_i x_i$, we exploit the identity $S_x = S_{x^2} + S_{x-x^2}$ (since $x = x^2 + (x - x^2)$) and the constraint $S_x = n - S_{1-x}$:

\begin{proposition}[Two Estimators for $S_x$]
\label{prop:estimators_x}
The estimators $\hat{S}_x^{(1)} := \tilde{S}_{x^2} + \tilde{S}_{x-x^2}$ and $\hat{S}_x^{(2)} := \tilde{n}_y - \tilde{S}_{1-x}$ are independent unbiased estimators of $S_x$ with variances $16/\varepsilon^2$ and $32/\varepsilon^2$, respectively.
\end{proposition}

The first estimator uses two Group 1 noises $(Z_{11}, Z_{12})$; the second uses $(Z_{13}, Z_{21}, Z_{22}, Z_{23})$---disjoint sets, ensuring independence.
Applying Lemma~\ref{lem:optimal_weights}, the optimal weights are $2/3$ and $1/3$:

\begin{theorem}[Refined Estimator for $S_x$]
\label{thm:refined_x}
Define:
$$\hat{S}_x := \frac{2}{3}\bigl(\tilde{S}_{x^2} + \tilde{S}_{x-x^2}\bigr) + \frac{1}{3}\bigl(\tilde{n}_y - \tilde{S}_{1-x}\bigr).$$
Then $\hat{S}_x$ is unbiased with variance $\mathrm{Var}(\hat{S}_x) = \dfrac{32}{3\varepsilon^2} \approx \dfrac{10.67}{\varepsilon^2}$.
\end{theorem}

This is a $33\%$ improvement over the first estimator alone.

\subsubsection{Refined Estimator for $S_y$}

Similarly, $S_y = S_{xy} + S_{(1-x)y}$ (since $y = xy + (1-x)y$) and $S_y = n - S_{1-y}$:

\begin{proposition}[Two Estimators for $S_y$]
\label{prop:estimators_y}
The estimators $\hat{S}_y^{(1)} := \tilde{S}_{xy} + \tilde{S}_{(1-x)y}$ and $\hat{S}_y^{(2)} := \tilde{n}_x - \tilde{S}_{1-y}$ are independent unbiased estimators of $S_y$ with variances $16/\varepsilon^2$ and $32/\varepsilon^2$, respectively.
\end{proposition}

Applying Lemma~\ref{lem:optimal_weights}, the optimal weights are $2/3$ and $1/3$:

\begin{theorem}[Refined Estimator for $S_y$]
\label{thm:refined_y}
Define:
$$\hat{S}_y := \frac{2}{3}\bigl(\tilde{S}_{xy} + \tilde{S}_{(1-x)y}\bigr) + \frac{1}{3}\bigl(\tilde{n}_x - \tilde{S}_{1-y}\bigr).$$
Then $\hat{S}_y$ is unbiased with variance $\mathrm{Var}(\hat{S}_y) = \dfrac{32}{3\varepsilon^2}$.
\end{theorem}

\subsection{Summary of Variance Reduction}

\begin{table}[h]
\centering
\caption{Comparison of variances: DP-RSS vs.\ DP-SS}
\label{tab:variance_summary}
\begin{tabular}{lccc}
\toprule
\textbf{Statistic} & \textbf{DP-SS} $\mathrm{Var}(\tilde{\cdot})$ & \textbf{DP-RSS} $\mathrm{Var}(\hat{\cdot})$ & \textbf{Improv} \\
\midrule
$n$        & $\dfrac{16}{\varepsilon^2}$ & $\dfrac{12}{\varepsilon^2}$  & $1.33\times$ \\[5pt]
$S_{x^2}$  & $\dfrac{32}{\varepsilon^2}$ & $\dfrac{20}{3\varepsilon^2}$ & $4.80\times$ \\[5pt]
$S_{xy}$   & $\dfrac{32}{\varepsilon^2}$ & $\dfrac{20}{3\varepsilon^2}$ & $4.80\times$ \\[5pt]
$S_x$      & $\dfrac{32}{\varepsilon^2}$ & $\dfrac{32}{3\varepsilon^2}$ & $3.00\times$ \\[5pt]
$S_y$      & $\dfrac{32}{\varepsilon^2}$ & $\dfrac{32}{3\varepsilon^2}$ & $3.00\times$ \\[5pt]
\bottomrule
\end{tabular}
\end{table}

\paragraph{Variance Reduction.}
Table~\ref{tab:variance_summary} compares the variances of the DP-SS noisy estimators $\tilde{\cdot}$ against the DP-RSS refined estimators $\hat{\cdot}$. See Appendix~\ref{app:proofs} for a detailed derivation of these variances and optimal weights. We remark that DP-RSS can be readily extended to private polynomial regression (see Section~\ref{sec:polynomial}), and that the application of our simplex transformation technique to other problem scenarios may be of independent interest.

The most dramatic improvements occur for $S_{x^2}$ and $S_{xy}$ --- these are the \emph{quadratic} statistics that are most critical for linear regression. The variance of the estimated slope $\hat{\alpha}$ depends heavily on $\mathrm{Var}(S_{x^2})$ and $\mathrm{Var}(S_{xy})$, so nearly $5\times$ reduction in these variances translates to significantly more accurate regression coefficients.

\section{Experimental Results and Discussion}

\subsection{Experimental Setup}

\subsubsection{Data Generation}
We generated two synthetic datasets to evaluate our method. In both setups, the independent variable $x_i$ was sampled from a Uniform distribution $x_i \in [0, 1]$, and the dependent variable $y_i$ was generated according to the linear model $y_i = \alpha x_i + \beta + e_i$, where $e_i \sim \mathcal{N}(0, \sigma^2)$ is Gaussian noise. The resulting $y_i$ values were clipped to $[0, 1]$ to satisfy the boundedness constraint required by the privacy mechanisms. The specific parameters for each setup are summarized in Table~\ref{tab:setups}.

\begin{table}[h]
\centering
\caption{Parameters for synthetic data generation.}
\label{tab:setups}
\begin{tabular}{lcccc}
\toprule
& $n$ & $\alpha$ & $\beta$ & $\sigma$ \\
\midrule
Setup 1 & 5{,}000  & $-0.7$ & $0.8$ & $0.05$ \\
Setup 2 & 10{,}000 & $0.5$  & $0.2$ & $0.1$  \\
\bottomrule
\end{tabular}
\end{table}

\begin{figure}[t]
\centering
\includegraphics[width=0.49\textwidth]{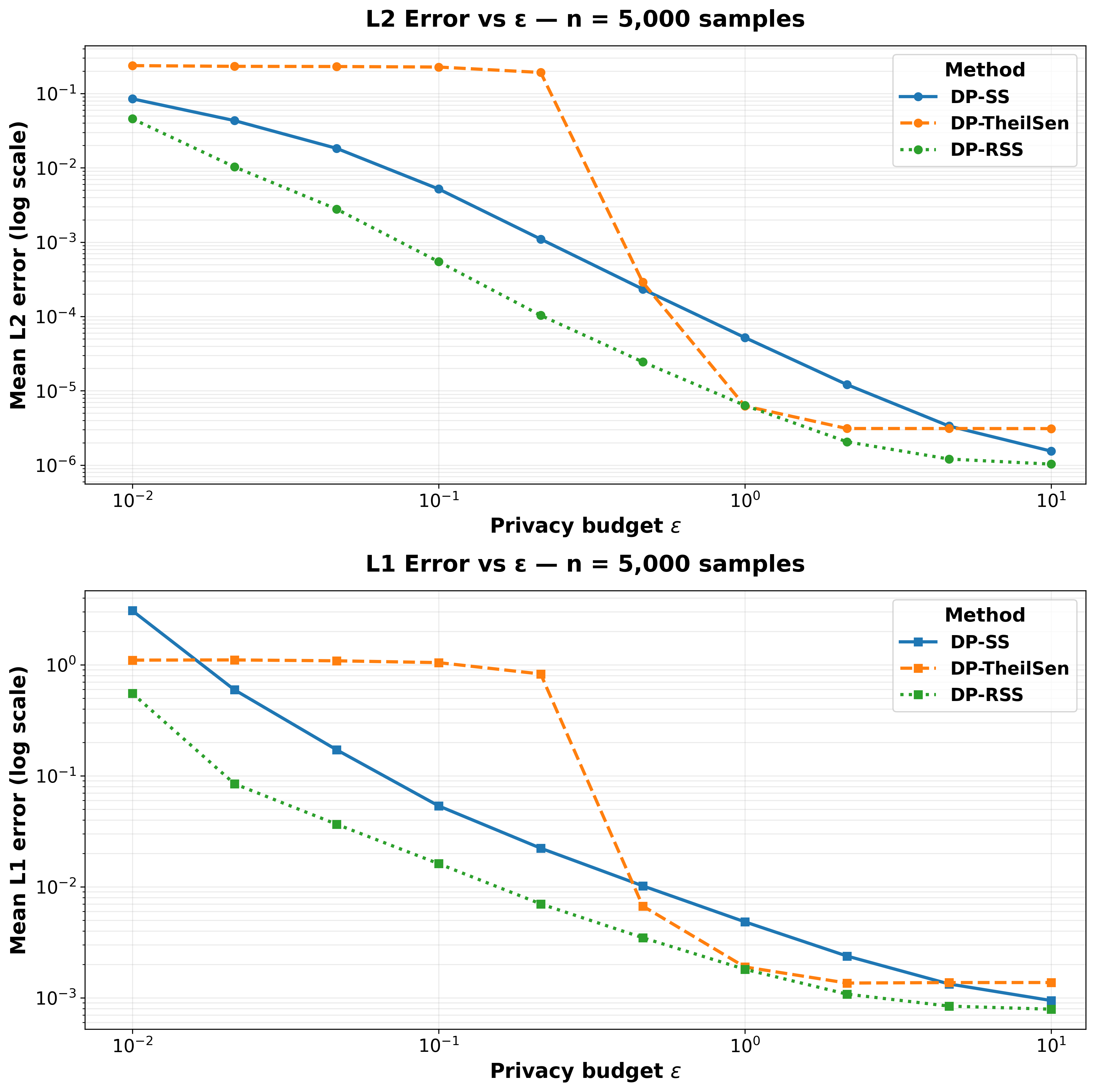}
\hfill
\includegraphics[width=0.49\textwidth]{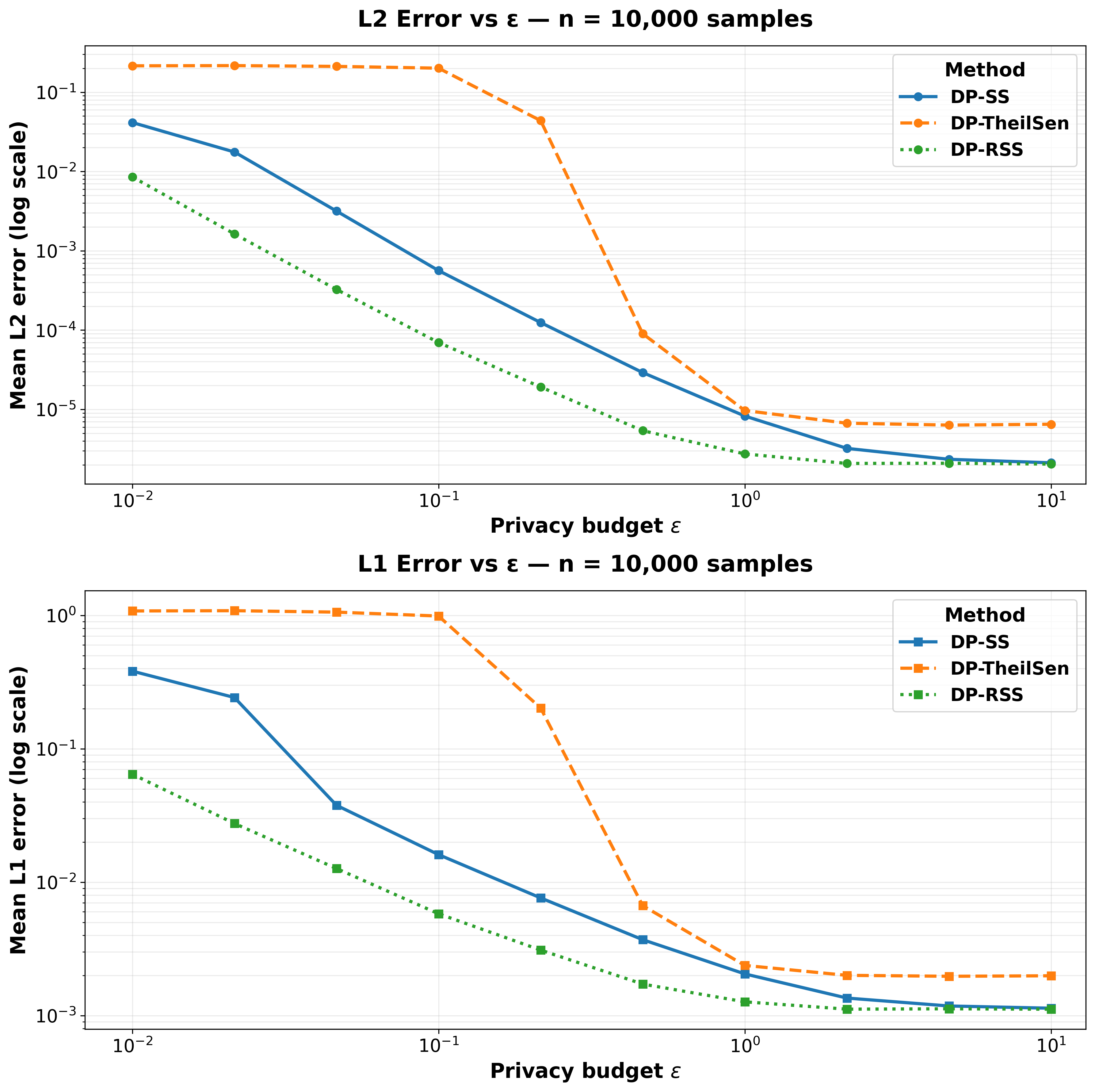}
\caption{Comparison of L1 and L2 errors versus privacy budget $\varepsilon$ for different DP regression methods. The sub-plots on the left (resp.\ right) correspond to Setup~$1$ (resp.\ Setup~$2$).}
\label{fig:errors_5k_wide}
\end{figure}

\subsubsection{Evaluation Metrics}
The evaluation metrics used were the $L_1$ Error~\eqref{eq:mae}, and the $L_2$ Error~\eqref{eq:mse}, of the estimated regression lines relative to the true regression line, where the expectation in~\eqref{eq:mae} and~\eqref{eq:mse} was taken over a Uniform distribution of $x \in [0, 1]$. For $L_1$ error, we approximated the integral using numerical integration over $n_{\text{points}} = 1{,}000$ equally spaced points:
\begin{equation}
L_1 \approx \frac{1}{n_{\text{points}}} \sum_{i=1}^{n_{\text{points}}} |(\alpha x_i + \beta) - (\hat{\alpha} x_i + \hat{\beta})|
\label{eq:mae1}
\end{equation}
where $x_i = i/n_{\text{points}}$. For $L_2$ error, we computed the exact integral analytically. Let $A := \alpha - \hat{\alpha}$, and $B := \beta - \hat{\beta}$. Given the difference function $h(x) = Ax + B$, the mean squared error is:
\begin{equation}
L_2 = \int_0^1 (Ax + B)^2 \, dx = \frac{A^2}{3} + AB + B^2.
\label{eq:mse1}
\end{equation}
Both metrics measure the integrated discrepancy between the true and estimated regression functions over the entire input domain, rather than only at the training points. The experiments were repeated over 1{,}000 iterations for each $\varepsilon$ to ensure statistical significance.

\subsection{Results}

We evaluate the performance of our proposed method, DP-RSS, by comparing it against two baselines: the standard Differentially Private Sufficient Statistics (DP-SS) \citep{alabi2022differentially} and the DP-Theil-Sen estimator \citep{alabi2022differentially,dwork2009RobustStatistics}. The evaluation metrics are the $L_1$ and $L_2$ estimation errors of the regression lines given by \eqref{eq:mae1} and \eqref{eq:mse1}, respectively. As illustrated in Figure~\ref{fig:errors_5k_wide}, both $L_1$ and $L_2$ errors decrease as $\varepsilon$ increases, with the most significant reductions observed in the low-$\varepsilon$ regime. Our DP-RSS method consistently outperforms the baseline DP-SS technique in both the setups, corroborating our analytical results in Table~\ref{tab:variance_summary} on variance reduction with DP-RSS. Further, for $\varepsilon < 1.0$, Figure~\ref{fig:errors_5k_wide} demonstrates that DP-RSS gives many-fold reduction in error over DP-Theil-Sen.

\section{Extension to Polynomial Regression}
\label{sec:polynomial}

Our simplex transformation approach generalizes naturally to polynomial regression of degree $d$. For polynomial $f(X) = \sum_{i=0}^d a_i X^i$, we construct two groups of statistics with unit $\ell_1$-sensitivity:
\begin{itemize}
  \item \textbf{Group 1:} $[S_{x^{2d}}, S_{x^{2d-1}-x^{2d}}, \ldots, S_{x-x^2}, S_{1-x}]$, representing $x$-\emph{statistics} of dimension $2d+1$
  \item \textbf{Group 2:} $[S_{x^d y}, S_{(x^{d-1}-x^d)y}, \ldots, S_{(1-x)y}, S_{1-y}]$, representing \emph{joint statistics} of dimension $d+2$
\end{itemize}
Note that the vectors in the above groups have sensitivity $1$ for the add/remove DP model. Using privacy budgets $\varepsilon_1, \varepsilon_2$ for Groups 1 and 2, respectively, we obtain noisy statistics:
\begin{align*}
\tilde{S}_{x^{2d}} &= S_{x^{2d}} + \mathrm{Lap}(1/\varepsilon_1), \\
\tilde{S}_{x^{j-1}-x^j} &= S_{x^{j-1}-x^j} + \mathrm{Lap}(1/\varepsilon_1), \quad j = 2d, \ldots, 1, \\
\tilde{S}_{x^d y} &= S_{x^d y} + \mathrm{Lap}(1/\varepsilon_2), \qquad \tilde{S}_{1-y} = S_{1-y} + \mathrm{Lap}(1/\varepsilon_2), \\
\tilde{S}_{(x^{l-1} - x^l) y} &= S_{(x^{l-1} - x^l) y} + \mathrm{Lap}(1/\varepsilon_2), \quad l = d, \ldots, 1.
\end{align*}

The above noisy statistics are $(\varepsilon_1 + \varepsilon_2)$-DP, and post-processing yields refined private sufficient statistics: $\hat{S}_{x^j}$ for $1 \le j \le 2d$, and $\hat{S}_{x^l y}$ for $0 \le l \le d$, via inverse-variance weighting.
\begin{align*}
\tilde{n}_x &:= \tilde{S}_{x^{2d}} + \tilde{S}_{x^{2d-1}-x^{2d}} + \cdots + \tilde{S}_{x-x^{2}} + \tilde{S}_{1-x}, \\
\tilde{n}_y &:= \tilde{S}_{x^{d}y} + \tilde{S}_{(x^{d-1}-x^{d})y} + \cdots + \tilde{S}_{(1-x)y} + \tilde{S}_{1-y}, \\
\hat{n} &= w_1\tilde{n}_x + w_2\tilde{n}_y, \\
\hat{S}_{x^{l}} &= w_1\!\left[ \tilde{S}_{x^{2d}} + \sum_{j=l+1}^{2d} \tilde{S}_{x^{j-1}-x^{j}} \right] + w_2\!\left[ \tilde{n}_y - \sum_{j=1}^{l} \tilde{S}_{x^{j-1}-x^{j}} \right], \quad l = 1, \ldots, 2d, \\
\hat{S}_{x^{l}y} &= w_1\!\left[ \tilde{S}_{x^{d}y} + \sum_{j=l+1}^{d} \tilde{S}_{(x^{j-1}-x^{j})y} \right] + w_2\!\left[ \tilde{n}_x - \sum_{j=1}^{l} \tilde{S}_{(x^{j-1}-x^{j})y} - \tilde{S}_{1-y} \right], \quad l = 0, \ldots, d.
\end{align*}

A private estimate of the polynomial model $\widehat{\mathbf{a}} = [\hat{a}_d, \ldots, \hat{a}_0]^\top$ is recovered via the normal equations:
\[
\widehat{\mathbf{a}} = \widehat{\mathbf{X}}^{-1} \widehat{\mathbf{y}}, \quad \text{with} \quad \hat{X}_{i,j} = \hat{S}_{x^{2d-i-j}},\quad \hat{y}_i = \hat{S}_{x^{d-i}y}, \quad 0 \le i, j \le d.
\]

\section{Conclusion}

In this work, we introduced DP-RSS, a novel mechanism for differentially private linear regression in the add/remove DP model. By exploiting the algebraic structure of bounded data and applying novel simplex transformations, we derived multiple independent estimators for key sufficient statistics without consuming additional privacy budget. Our theoretical analysis and empirical evaluations demonstrate that DP-RSS consistently outperforms the standard DP-SS estimator. Further, compared to the DP-Theil-Sen estimator, our experimental results highlight a significant reduction in error with our DP-RSS mechanism in the high-privacy (low $\varepsilon$) regime. Finally, we remark that our multidimensional simplex transformation can be readily applied to the adaptive framework presented in \citet[][Alg.~2]{wang2018PrivateLinearRegression}, and can further improve the performance of the BinAgg algorithm \citep{lin2026LinearRegression} via refined estimates of private sums.

\paragraph{Future Work.}
Several promising directions remain open: (1) \emph{Optimal privacy allocation for polynomial regression}: While our approach naturally extends to degree-$d$ polynomials, determining the optimal split of privacy budget $\varepsilon = \varepsilon_1 + \varepsilon_2$ among the $O(d)$ statistics to minimize regression error is an important open problem. (2) \emph{Multivariate regression}: Extending our simplex transformation techniques to multiple input variables while maintaining computational efficiency. (3) \emph{Other statistical tasks}: Exploring whether similar algebraic constraints can be leveraged for variance reduction in other estimation problems under add/remove differential privacy, such as generalized linear models or robust regression.

\bibliography{sample-base}

\appendix

\section{DP-RSS for General Bounded Data}
\label{app:general_bounds}

Algorithm~\ref{alg:dp_rss_compact} and the theoretical results in Section~4 are stated for data satisfying $x_i, y_i \in [0,1]$. In practice, however, the raw data may only be known to lie in a general rectangle $[x_{\min}, x_{\max}] \times [y_{\min}, y_{\max}]$. We now describe how DP-RSS adapts to this setting via a \emph{linear normalisation}, which reduces the general case exactly to the $[0,1]^2$ case analysed in the main paper, and then show how to recover regression parameters on the original scale.

Given a dataset $\mathcal{D} = \{(x_i, y_i)\}_{i=1}^n$ with $x_{\min} \leq x_i \leq x_{\max}$ and $y_{\min} \leq y_i \leq y_{\max}$, define the normalised variables
\[
x_i' = \frac{x_i - x_{\min}}{x_{\max} - x_{\min}}, \qquad y_i' = \frac{y_i - y_{\min}}{y_{\max} - y_{\min}},
\]
so that $x_i', y_i' \in [0,1]$ for all $i$. Let $\Delta_x := x_{\max} - x_{\min} > 0$ and $\Delta_y := y_{\max} - y_{\min} > 0$. The normalised dataset $\mathcal{D}' = \{(x_i', y_i')\}_{i=1}^n$ satisfies the boundedness assumption of Algorithm~\ref{alg:dp_rss_compact}, and DP-RSS can be applied to $\mathcal{D}'$ directly to obtain private estimates $\hat{\alpha}'$ (slope) and $\hat{\beta}'$ (intercept) of the linear model $y_i' = \alpha' x_i' + \beta' + e_i$, while satisfying $\varepsilon$-differential privacy with the same privacy guarantees established in Theorem~\ref{thm:privacy_composition}.

The estimated regression line on the original scale is obtained by inverting the normalisation above. Substituting $y_i' = (y_i - y_{\min})/\Delta_y$ and $x_i' = (x_i - x_{\min})/\Delta_x$ into the normalised model and rearranging gives
$y_i = y_{\min} + \Delta_y \!\left(\alpha' \cdot \frac{x_i - x_{\min}}{\Delta_x} + \beta'\right).$
Hence, the private estimates on the original scale are:
\begin{align*}
\hat{\alpha} &= \frac{\Delta_y}{\Delta_x}\,\hat{\alpha}', \\
\hat{\beta}  &= y_{\min} + \Delta_y\!\left(\hat{\beta}' - \frac{x_{\min}}{\Delta_x}\,\hat{\alpha}'\right).
\end{align*}

\section{Proofs of Main Results}
\label{app:proofs}

\subsection{Sensitivity and Privacy Proofs}

\begin{proof}[Proof of Lemma~\ref{lem:sensitivity_group1}]
Consider two neighboring datasets $D$ and $D'$ under the add/remove model (unbounded differential privacy). Without loss of generality, assume that $D'$ contains one additional record compared to $D$:
$D' = D \cup \{(x^*, y^*)\}$
where $(x^*, y^*) \in [0,1] \times [0,1]$ is the added record.

We compute each component of $f(D') - f(D)$ separately.
\begin{align*}
S_{x^2}(D') &= \sum_{(x_i, y_i) \in D'} x_i^2 = \sum_{(x_i, y_i) \in D} x_i^2 + (x^*)^2 = S_{x^2}(D) + (x^*)^2
\end{align*}
\begin{align*}
S_{x-x^2}(D') &= \sum_{(x_i, y_i) \in D'} (x_i - x_i^2) = \sum_{(x_i, y_i) \in D} (x_i - x_i^2) + (x^* - (x^*)^2) \\
&= S_{x-x^2}(D) + (x^* - (x^*)^2)
\end{align*}
\begin{align*}
S_{1-x}(D') &= \sum_{(x_i, y_i) \in D'} (1 - x_i) = \sum_{(x_i, y_i) \in D} (1 - x_i) + (1 - x^*) \\
&= S_{1-x}(D) + (1 - x^*)
\end{align*}

Therefore, the difference vector is:
$f(D') - f(D) = \bigl( (x^*)^2,\; x^* - (x^*)^2,\; 1 - x^* \bigr).$

Since $x^* \in [0,1]$, each component is non-negative.
\begin{align*}
\|f(D') - f(D)\|_1 &= |(x^*)^2| + |x^* - (x^*)^2| + |1 - x^*| \\
&= (x^*)^2 + (x^* - (x^*)^2) + (1 - x^*) = 1.
\end{align*}

Since the $\ell_1$-norm equals 1 for any choice of $x^* \in [0,1]$, the global $\ell_1$-sensitivity is:
$\Delta_1(f) = \max_{D \sim D'} \|f(D) - f(D')\|_1 = 1.$
\end{proof}

\begin{proof}[Proof of Lemma~\ref{lem:sensitivity_group2}]
Consider neighboring datasets $D$ and $D'$ where $D' = D \cup \{(x^*, y^*)\}$ with $(x^*, y^*) \in [0,1]^2$.

We compute each component of $g(D') - g(D)$ separately.
\begin{align*}
S_{xy}(D') &= \sum_{(x_i, y_i) \in D'} x_i y_i = S_{xy}(D) + x^* y^*
\end{align*}
\begin{align*}
S_{(1-x)y}(D') &= \sum_{(x_i, y_i) \in D'} (1-x_i) y_i = S_{(1-x)y}(D) + (1-x^*) y^*
\end{align*}
\begin{align*}
S_{1-y}(D') &= \sum_{(x_i, y_i) \in D'} (1 - y_i) = S_{1-y}(D) + (1 - y^*)
\end{align*}

Therefore, the difference vector is:
$g(D') - g(D) = \bigl( x^* y^*,\; (1-x^*)y^*,\; 1-y^* \bigr).$

Since $x^*, y^* \in [0,1]$, all three components are non-negative.
\begin{align*}
\|g(D') - g(D)\|_1 &= x^* y^* + (1-x^*)y^* + (1-y^*) = 1.
\end{align*}

Since the $\ell_1$-norm equals 1 for any choice of $(x^*, y^*) \in [0,1]^2$, the global $\ell_1$-sensitivity is:
$\Delta_1(g) = \max_{D \sim D'} \|g(D) - g(D')\|_1 = 1.$
\end{proof}

\begin{proof}[Proof of Corollary~\ref{cor:privacy_group1}]
We apply the Laplace mechanism theorem for vector-valued queries.

\textbf{Statement of the Laplace mechanism.} Let $f: \mathcal{D} \to \mathbb{R}^k$ be a vector-valued function with global $\ell_1$-sensitivity $\Delta_1(f)$. The mechanism that releases
$$\tilde{f}(D) = f(D) + (Z_1, Z_2, \ldots, Z_k)$$
where $Z_1, \ldots, Z_k \stackrel{\text{iid}}{\sim} \mathrm{Lap}(\Delta_1(f)/\varepsilon)$, satisfies $\varepsilon$-differential privacy.

By Lemma~\ref{lem:sensitivity_group1}, the function $f(D) = (S_{x^2}, S_{x-x^2}, S_{1-x})$ has $\ell_1$-sensitivity $\Delta_1(f) = 1$.

We release:
$(\tilde{S}_{x^2}, \tilde{S}_{x-x^2}, \tilde{S}_{1-x}) = (S_{x^2} + Z_{11}, S_{x-x^2} + Z_{12}, S_{1-x} + Z_{13})$
where $Z_{11}, Z_{12}, Z_{13} \stackrel{\text{iid}}{\sim} \mathrm{Lap}(1/\varepsilon_1)$.

Since the noise scale is $b = \Delta_1(f)/\varepsilon_1 = 1/\varepsilon_1$, the Laplace mechanism theorem guarantees $\varepsilon_1$-differential privacy.

\textbf{Formal verification.} For any neighboring datasets $D, D'$ and any measurable set $S \subseteq \mathbb{R}^3$:
\begin{align*}
\frac{\Pr[(\tilde{S}_{x^2}, \tilde{S}_{x-x^2}, \tilde{S}_{1-x}) \in S \mid D]}{\Pr[(\tilde{S}_{x^2}, \tilde{S}_{x-x^2}, \tilde{S}_{1-x}) \in S \mid D']} &\leq \exp\!\left(\varepsilon_1 \cdot \|f(D) - f(D')\|_1\right) \\
&\leq \exp\!\left(\varepsilon_1 \cdot \Delta_1(f)\right) = e^{\varepsilon_1}.
\end{align*}
This is the definition of $\varepsilon_1$-differential privacy (with $\delta = 0$).
\end{proof}

\begin{proof}[Proof of Corollary~\ref{cor:privacy_group2}]
The proof is completely analogous to Corollary~\ref{cor:privacy_group1}.

By Lemma~\ref{lem:sensitivity_group2}, the function $g(D) = (S_{xy}, S_{(1-x)y}, S_{1-y})$ has $\ell_1$-sensitivity $\Delta_1(g) = 1$.

We release:
$(\tilde{S}_{xy}, \tilde{S}_{(1-x)y}, \tilde{S}_{1-y}) = (S_{xy} + Z_{21}, S_{(1-x)y} + Z_{22}, S_{1-y} + Z_{23})$
where $Z_{21}, Z_{22}, Z_{23} \stackrel{\text{iid}}{\sim} \mathrm{Lap}(1/\varepsilon_2)$.

Since the noise scale is $b = \Delta_1(g)/\varepsilon_2 = 1/\varepsilon_2$, the Laplace mechanism theorem guarantees $\varepsilon_2$-differential privacy.
\end{proof}

\begin{proof}[Proof of Theorem~\ref{thm:privacy_composition}]
We apply the basic composition theorem for pure differential privacy.

\textbf{Statement of the composition theorem.} If $\mathcal{M}_1: \mathcal{D} \to \mathcal{R}_1$ satisfies $\varepsilon_1$-DP and $\mathcal{M}_2: \mathcal{D} \to \mathcal{R}_2$ satisfies $\varepsilon_2$-DP, then the combined mechanism $\mathcal{M}(D) = (\mathcal{M}_1(D), \mathcal{M}_2(D))$ satisfies $(\varepsilon_1 + \varepsilon_2)$-DP.

Define:
\begin{itemize}
\item $\mathcal{M}_1(D) = (\tilde{S}_{x^2}, \tilde{S}_{x-x^2}, \tilde{S}_{1-x})$, which satisfies $\varepsilon_1$-DP by Corollary~\ref{cor:privacy_group1}.
\item $\mathcal{M}_2(D) = (\tilde{S}_{xy}, \tilde{S}_{(1-x)y}, \tilde{S}_{1-y})$, which satisfies $\varepsilon_2$-DP by Corollary~\ref{cor:privacy_group2}.
\end{itemize}

The joint mechanism is:
$$\mathcal{M}(D) = (\mathcal{M}_1(D), \mathcal{M}_2(D)) = \bigl( \tilde{S}_{x^2}, \tilde{S}_{x-x^2}, \tilde{S}_{1-x}, \tilde{S}_{xy}, \tilde{S}_{(1-x)y}, \tilde{S}_{1-y} \bigr).$$

By the composition theorem, $\mathcal{M}(D)$ satisfies $(\varepsilon_1 + \varepsilon_2)$-differential privacy. With $\varepsilon_1 = \varepsilon_2 = \varepsilon/2$, the total privacy budget is $\varepsilon_1 + \varepsilon_2 = \varepsilon$.
\end{proof}

\subsection{Proofs for Estimator Properties}

\begin{proof}[Proof of Lemma~\ref{lem:optimal_weights}]
We seek to minimize the variance of a weighted combination of two independent unbiased estimators.

Let $\hat{\theta}_1$ and $\hat{\theta}_2$ be independent unbiased estimators of $\theta$ with variances $\sigma_1^2$ and $\sigma_2^2$, respectively. Consider the weighted average $\hat{\theta} = w_1 \hat{\theta}_1 + w_2 \hat{\theta}_2$ subject to the constraint $w_1 + w_2 = 1$.

By independence:
$\mathrm{Var}(\hat{\theta}) = w_1^2 \sigma_1^2 + w_2^2 \sigma_2^2.$

Using the constraint $w_2 = 1 - w_1$:
\begin{align*}
V(w_1) &= w_1^2 \sigma_1^2 + (1-w_1)^2 \sigma_2^2 = w_1^2 (\sigma_1^2 + \sigma_2^2) - 2w_1\sigma_2^2 + \sigma_2^2.
\end{align*}

Taking the derivative and setting to zero:
$\frac{dV}{dw_1} = 2w_1(\sigma_1^2 + \sigma_2^2) - 2\sigma_2^2 = 0 \implies w_1^* = \frac{\sigma_2^2}{\sigma_1^2 + \sigma_2^2}.$

Therefore $w_2^* = \frac{\sigma_1^2}{\sigma_1^2 + \sigma_2^2}$. Since $\frac{d^2V}{dw_1^2} = 2(\sigma_1^2 + \sigma_2^2) > 0$, this is a minimum.

Substituting the optimal weights:
$$\mathrm{Var}(\hat{\theta}^*) = \left(\frac{\sigma_2^2}{\sigma_1^2 + \sigma_2^2}\right)^2 \sigma_1^2 + \left(\frac{\sigma_1^2}{\sigma_1^2 + \sigma_2^2}\right)^2 \sigma_2^2 = \frac{\sigma_1^2 \sigma_2^2}{\sigma_1^2 + \sigma_2^2}.$$
\end{proof}

\begin{proof}[Proof of Proposition~\ref{prop:estimators_x2}]
We prove independence and compute variances for both estimators.

\textbf{Part 1: First estimator $\hat{S}_{x^2}^{(1)} = \tilde{S}_{x^2}$.}

\textit{Variance:} For $Z \sim \mathrm{Lap}(b)$, $\mathrm{Var}(Z) = 2b^2$. With $b = 1/\varepsilon_1 = 2/\varepsilon$:
$$\mathrm{Var}\!\left(\hat{S}_{x^2}^{(1)}\right) = \mathrm{Var}(Z_{11}) = 2\left(\frac{2}{\varepsilon}\right)^2 = \frac{8}{\varepsilon^2}.$$

\textbf{Part 2: Second estimator $\hat{S}_{x^2}^{(2)} = \tilde{n}_y - (\tilde{S}_{x-x^2} + \tilde{S}_{1-x})$.}

\textit{Expanding the definition:}
\begin{align*}
\hat{S}_{x^2}^{(2)} &= (S_{xy} + Z_{21}) + (S_{(1-x)y} + Z_{22}) + (S_{1-y} + Z_{23}) \\
&\quad - (S_{x-x^2} + Z_{12}) - (S_{1-x} + Z_{13}) \\
&= \underbrace{S_{xy} + S_{(1-x)y} + S_{1-y} - S_{x-x^2} - S_{1-x}}_{\text{deterministic part}} + \underbrace{Z_{21} + Z_{22} + Z_{23} - Z_{12} - Z_{13}}_{\text{noise part}}.
\end{align*}

\textit{Simplifying the deterministic part:} Using $S_{xy} + S_{(1-x)y} + S_{1-y} = n$ and $S_{x^2} + S_{x-x^2} + S_{1-x} = n$:
$$n - S_{x-x^2} - S_{1-x} = S_{x^2} \implies \hat{S}_{x^2}^{(2)} = S_{x^2} + Z_{21} + Z_{22} + Z_{23} - Z_{12} - Z_{13}.$$

\textit{Variance:} All noise variables have variance $8/\varepsilon^2$:
$$\mathrm{Var}(\hat{S}_{x^2}^{(2)}) = 5 \cdot \frac{8}{\varepsilon^2} = \frac{40}{\varepsilon^2}.$$

\textbf{Part 3: Independence.} $\hat{S}_{x^2}^{(1)}$ depends on $\{Z_{11}\}$ and $\hat{S}_{x^2}^{(2)}$ depends on $\{Z_{12}, Z_{13}, Z_{21}, Z_{22}, Z_{23}\}$. These are disjoint sets of independent noise variables, so the estimators are independent.
\end{proof}

\begin{proof}[Proof of Theorem~\ref{thm:refined_x2}]
We apply Lemma~\ref{lem:optimal_weights} with $\sigma_1^2 = 8/\varepsilon^2$ and $\sigma_2^2 = 40/\varepsilon^2$.

Computing optimal weights: By Lemma~\ref{lem:optimal_weights}:
$$w_1 = \frac{\sigma_2^2}{\sigma_1^2 + \sigma_2^2} = \frac{40/\varepsilon^2}{8/\varepsilon^2 + 40/\varepsilon^2} = \frac{5}{6}, \qquad w_2 = 1 - w_1 = \frac{1}{6}.$$

Refined estimator:
$$\hat{S}_{x^2} = w_1 \hat{S}_{x^2}^{(1)} + w_2 \hat{S}_{x^2}^{(2)} = \frac{5}{6}\,\tilde{S}_{x^2} + \frac{1}{6}\bigl[\tilde{n}_y - (\tilde{S}_{x-x^2} + \tilde{S}_{1-x})\bigr].$$

Computing the variance: By Lemma~\ref{lem:optimal_weights}, the minimum variance is:
$$\mathrm{Var}(\hat{S}_{x^2}) = \frac{\sigma_1^2 \sigma_2^2}{\sigma_1^2 + \sigma_2^2} = \frac{(8/\varepsilon^2)(40/\varepsilon^2)}{8/\varepsilon^2 + 40/\varepsilon^2} = \frac{20}{3\varepsilon^2}.$$
\end{proof}

\begin{proof}[Proof of Proposition~\ref{prop:estimators_xy}]
We prove independence and compute variances for both estimators of $S_{xy}$.

\textbf{Part 1: First estimator $\hat{S}_{xy}^{(1)} = \tilde{S}_{xy}$.}
$$\mathrm{Var}(\hat{S}_{xy}^{(1)}) = \mathrm{Var}(Z_{21}) = 2 \cdot \left(\frac{2}{\varepsilon}\right)^2 = \frac{8}{\varepsilon^2}.$$

\textbf{Part 2: Second estimator $\hat{S}_{xy}^{(2)} = \tilde{n}_x - (\tilde{S}_{(1-x)y} + \tilde{S}_{1-y})$.}

\textit{Expanding the definition:}
\begin{align*}
\hat{S}_{xy}^{(2)} &= (S_{x^2} + Z_{11}) + (S_{x-x^2} + Z_{12}) + (S_{1-x} + Z_{13}) \\
&\quad - (S_{(1-x)y} + Z_{22}) - (S_{1-y} + Z_{23}) \\
&= \underbrace{S_{x^2} + S_{x-x^2} + S_{1-x} - S_{(1-x)y} - S_{1-y}}_{\text{deterministic part}} + \underbrace{Z_{11} + Z_{12} + Z_{13} - Z_{22} - Z_{23}}_{\text{noise part}}.
\end{align*}

\textit{Simplifying the deterministic part:} Using $S_{x^2} + S_{x-x^2} + S_{1-x} = n$ and $S_{xy} + S_{(1-x)y} + S_{1-y} = n$:
$$n - S_{(1-x)y} - S_{1-y} = S_{xy} \implies \hat{S}_{xy}^{(2)} = S_{xy} + Z_{11} + Z_{12} + Z_{13} - Z_{22} - Z_{23}.$$

\textit{Variance:} $\mathrm{Var}(\hat{S}_{xy}^{(2)}) = 5 \cdot \frac{8}{\varepsilon^2} = \frac{40}{\varepsilon^2}.$

\textbf{Part 3: Independence.} $\hat{S}_{xy}^{(1)}$ depends on $\{Z_{21}\}$ and $\hat{S}_{xy}^{(2)}$ depends on $\{Z_{11}, Z_{12}, Z_{13}, Z_{22}, Z_{23}\}$. These are disjoint sets of independent noise variables, so the estimators are independent.
\end{proof}

\begin{proof}[Proof of Theorem~\ref{thm:refined_xy}]
We apply Lemma~\ref{lem:optimal_weights} with $\sigma_1^2 = 8/\varepsilon^2$ and $\sigma_2^2 = 40/\varepsilon^2$.

Computing optimal weights:
$$w_1 = \frac{40/\varepsilon^2}{8/\varepsilon^2 + 40/\varepsilon^2} = \frac{5}{6}, \qquad w_2 = \frac{1}{6}.$$

Refined estimator:
$$\hat{S}_{xy} = \frac{5}{6}\,\tilde{S}_{xy} + \frac{1}{6}\bigl[\tilde{n}_x - (\tilde{S}_{(1-x)y} + \tilde{S}_{1-y})\bigr].$$

Computing the variance:
$$\mathrm{Var}(\hat{S}_{xy}) = \frac{(8/\varepsilon^2)(40/\varepsilon^2)}{8/\varepsilon^2 + 40/\varepsilon^2} = \frac{20}{3\varepsilon^2}.$$
\end{proof}

\begin{proof}[Proof of Proposition~\ref{prop:estimators_x}]
We prove independence and compute variances for both estimators of $S_x = \sum_{i=1}^n x_i$.

\textbf{Part 1: First estimator $\hat{S}_x^{(1)} = \tilde{S}_{x^2} + \tilde{S}_{x-x^2}$.}

\textit{Relationship to $S_x$:} $S_{x^2} + S_{x-x^2} = \sum_{i=1}^n x_i^2 + \sum_{i=1}^n (x_i - x_i^2) = \sum_{i=1}^n x_i = S_x$.

Thus $\hat{S}_x^{(1)} = S_x + Z_{11} + Z_{12}$.

\textit{Variance:}
$$\mathrm{Var}(\hat{S}_x^{(1)}) = \mathrm{Var}(Z_{11}) + \mathrm{Var}(Z_{12}) = \frac{8}{\varepsilon^2} + \frac{8}{\varepsilon^2} = \frac{16}{\varepsilon^2}.$$

\textbf{Part 2: Second estimator $\hat{S}_x^{(2)} = \tilde{n}_y - \tilde{S}_{1-x}$.}

\textit{Expanding the definition:}
\begin{align*}
\hat{S}_x^{(2)} &= (S_{xy} + Z_{21}) + (S_{(1-x)y} + Z_{22}) + (S_{1-y} + Z_{23}) - (S_{1-x} + Z_{13}) \\
&= \underbrace{S_{xy} + S_{(1-x)y} + S_{1-y} - S_{1-x}}_{\text{deterministic part}} + \underbrace{Z_{21} + Z_{22} + Z_{23} - Z_{13}}_{\text{noise part}}.
\end{align*}

\textit{Simplifying the deterministic part:} Using $S_{xy} + S_{(1-x)y} + S_{1-y} = n$ and $S_x = n - S_{1-x}$:
$$n - S_{1-x} = S_x \implies \hat{S}_x^{(2)} = S_x + Z_{21} + Z_{22} + Z_{23} - Z_{13}.$$

\textit{Variance:}
$$\mathrm{Var}(\hat{S}_x^{(2)}) = 4 \cdot \frac{8}{\varepsilon^2} = \frac{32}{\varepsilon^2}.$$

\textbf{Part 3: Independence.} $\hat{S}_x^{(1)}$ depends on $\{Z_{11}, Z_{12}\}$ and $\hat{S}_x^{(2)}$ depends on $\{Z_{13}, Z_{21}, Z_{22}, Z_{23}\}$. These are disjoint sets of independent noise variables, so the estimators are independent.
\end{proof}

\begin{proof}[Proof of Theorem~\ref{thm:refined_x}]
We apply Lemma~\ref{lem:optimal_weights} with $\sigma_1^2 = 16/\varepsilon^2$ and $\sigma_2^2 = 32/\varepsilon^2$.

Computing optimal weights:
$$w_1' = \frac{32/\varepsilon^2}{16/\varepsilon^2 + 32/\varepsilon^2} = \frac{2}{3}, \qquad w_2' = \frac{1}{3}.$$

Refined estimator:
$$\hat{S}_x = \frac{2}{3}\bigl(\tilde{S}_{x^2} + \tilde{S}_{x-x^2}\bigr) + \frac{1}{3}\bigl(\tilde{n}_y - \tilde{S}_{1-x}\bigr).$$

Computing the variance:
$$\mathrm{Var}(\hat{S}_x) = \frac{(16/\varepsilon^2)(32/\varepsilon^2)}{16/\varepsilon^2 + 32/\varepsilon^2} = \frac{32}{3\varepsilon^2}.$$
\end{proof}

\begin{proof}[Proof of Proposition~\ref{prop:estimators_y}]
We prove independence and compute variances for both estimators of $S_y = \sum_{i=1}^n y_i$.

\textbf{Part 1: First estimator $\hat{S}_y^{(1)} = \tilde{S}_{xy} + \tilde{S}_{(1-x)y}$.}

\textit{Relationship to $S_y$:} $S_{xy} + S_{(1-x)y} = \sum_{i=1}^n x_i y_i + \sum_{i=1}^n (1-x_i) y_i = \sum_{i=1}^n y_i = S_y$.

Thus $\hat{S}_y^{(1)} = S_y + Z_{21} + Z_{22}$.

\textit{Variance:}
$$\mathrm{Var}(\hat{S}_y^{(1)}) = \mathrm{Var}(Z_{21}) + \mathrm{Var}(Z_{22}) = \frac{8}{\varepsilon^2} + \frac{8}{\varepsilon^2} = \frac{16}{\varepsilon^2}.$$

\textbf{Part 2: Second estimator $\hat{S}_y^{(2)} = \tilde{n}_x - \tilde{S}_{1-y}$.}

\textit{Expanding the definition:}
\begin{align*}
\hat{S}_y^{(2)} &= (S_{x^2} + Z_{11}) + (S_{x-x^2} + Z_{12}) + (S_{1-x} + Z_{13}) - (S_{1-y} + Z_{23}) \\
&= \underbrace{S_{x^2} + S_{x-x^2} + S_{1-x} - S_{1-y}}_{\text{deterministic part}} + \underbrace{Z_{11} + Z_{12} + Z_{13} - Z_{23}}_{\text{noise part}}.
\end{align*}

\textit{Simplifying the deterministic part:} Using $S_{x^2} + S_{x-x^2} + S_{1-x} = n$ and $S_y = n - S_{1-y}$:
$$n - S_{1-y} = S_y \implies \hat{S}_y^{(2)} = S_y + Z_{11} + Z_{12} + Z_{13} - Z_{23}.$$

\textit{Variance:}
$$\mathrm{Var}(\hat{S}_y^{(2)}) = 4 \cdot \frac{8}{\varepsilon^2} = \frac{32}{\varepsilon^2}.$$

\textbf{Part 3: Independence.} $\hat{S}_y^{(1)}$ depends on $\{Z_{21}, Z_{22}\}$ and $\hat{S}_y^{(2)}$ depends on $\{Z_{11}, Z_{12}, Z_{13}, Z_{23}\}$. These are disjoint sets of independent noise variables, so the estimators are independent.
\end{proof}

\begin{proof}[Proof of Theorem~\ref{thm:refined_y}]
We apply Lemma~\ref{lem:optimal_weights} with $\sigma_1^2 = 16/\varepsilon^2$ and $\sigma_2^2 = 32/\varepsilon^2$.

Computing optimal weights:
$$w_1' = \frac{32/\varepsilon^2}{16/\varepsilon^2 + 32/\varepsilon^2} = \frac{2}{3}, \qquad w_2' = \frac{1}{3}.$$

Refined estimator:
$$\hat{S}_y = \frac{2}{3}\bigl(\tilde{S}_{xy} + \tilde{S}_{(1-x)y}\bigr) + \frac{1}{3}\bigl(\tilde{n}_x - \tilde{S}_{1-y}\bigr).$$

Computing the variance:
$$\mathrm{Var}(\hat{S}_y) = \frac{(16/\varepsilon^2)(32/\varepsilon^2)}{16/\varepsilon^2 + 32/\varepsilon^2} = \frac{32}{3\varepsilon^2}.$$
\end{proof}

\end{document}